\documentclass{llncs}

\usepackage[T1]{fontenc}
\usepackage[utf8]{inputenc}
\usepackage{microtype}
\usepackage{hyperref}
\usepackage{amssymb,amsmath,enumerate}
\usepackage{comment}
\usepackage{tkz-graph,caption}

\newcommand{\NP}{{\sf NP}}

\newcommand{\ssi}{\subseteq_i}

\DeclareMathOperator{\cw}{cw}

\title{Bounding Clique-width via Perfect Graphs
\thanks{The research in this paper was supported by EPSRC (EP/K025090/1).
The second author is grateful for the generous support of the Graduate (International) Research Travel Award from Simon Fraser University and Dr. Pavol Hell's NSERC Discovery Grant.}}
\author{Konrad K. Dabrowski\inst{1} \and Shenwei Huang\inst{2} \and Dani\"el Paulusma\inst{1}}
\institute{
School of Engineering and  Computing Sciences, Durham University,\\
Science Laboratories, South Road, Durham DH1 3LE, United Kingdom
\texttt{\{konrad.dabrowski,daniel.paulusma\}@durham.ac.uk}
\and
School of Computing Science, Simon Fraser University,\\
8888 University Drive, Burnaby B.C., V5A 1S6, Canada\\
\texttt{shenweih@sfu.ca}
}

\begin{document}
\maketitle
\setcounter{footnote}{0}

\begin{abstract}
Given two graphs $H_1$ and $H_2$, a graph $G$ is $(H_1,H_2)$-free if it contains no subgraph isomorphic to $H_1$ or $H_2$.
We continue a  recent study into the clique-width of $(H_1,H_2)$-free graphs and  present three new classes of $(H_1,H_2)$-free graphs
of bounded clique-width
and one of unbounded clique-width.
The four new graph classes have in common that one of their two forbidden induced subgraphs is the diamond (the graph obtained from a clique on four vertices by deleting one edge).
To prove boundedness of clique-width for the first three cases we develop a technique  based on bounding clique covering number in combination with
reduction to subclasses of perfect graphs.
We extend our proof of unboundedness for the fourth case to show that {\sc Graph Isomorphism} is {\sc Graph Isomorphism}-complete on the same graph class.
We also show the implications of our results for the computational complexity of the {\sc Colouring} problem restricted to $(H_1,H_2)$-free graphs.

\end{abstract}

\section{Introduction}\label{s-intro}
Clique-width is a well-known graph parameter and its properties are well studied;
see for example the surveys of Gurski~\cite{Gu07} and Kami\'nski, Lozin and Milani\v{c}~\cite{KLM09}.
Computing the clique-width of a given graph is
\NP-hard, as shown by Fellows, Rosamond, Rotics and Szeider~\cite{FRRS09}.
Nevertheless, many \NP-complete graph problems are solvable in polynomial time on graph classes of {\it bounded} clique-width, that is, classes
in which the clique-width of each of its graphs is at most~$c$ for some constant~$c$.
This follows by combining the fact that if a graph $G$ has clique-width at most $c$ then a so-called $(8^c-1)$-expression for $G$ can be found in polynomial time~\cite{Oum08} together with a number of results~\cite{CMR00,KR03b,Ra07}, which show that if a $q$-expression is provided for some fixed $q$ then certain classes of problems can be solved in polynomial time.
A well-known example of such a problem is the {\sc Colouring} problem, which is that of testing whether the vertices of a graph can be coloured with at most $k$ colours such that no two adjacent vertices are coloured alike.  Due to these algorithmic implications, it is natural to research
whether the clique-width of a given graph class is bounded.

It should be noted that having bounded clique-width is a more general property than having bounded tree-width, that is, every graph class of bounded treewidth has bounded clique-width but the reverse is not true~\cite{CR05}.
Clique-width is also closely related to other graph width parameters, e.g. for any class, having bounded clique-width is equivalent to having bounded rank-width~\cite{OS06} and also equivalent to having bounded NLC-width~\cite{Johansson98}.
Moreover, clique-width has been studied in relation to graph operations, such as edge or vertex deletions, edge subdivisions and edge contractions.
For instance, a recent result of Courcelle~\cite{Co14} solved an open problem of Gurski~\cite{Gu07} by proving that if~${\cal G}$ is the class of graphs of clique-width~3 and~${\cal G}'$ is the class of graphs obtained from graphs in~${\cal G}$ by applying one or more edge contraction operations then~${\cal G}'$ has unbounded clique-width.

The classes that we consider in this paper consist of graphs that can be characterized by a family $\{H_1,\ldots,H_p\}$ of forbidden induced subgraphs (such graphs are
said to be $(H_1,\ldots,H_p)$-free).
The clique-width of such graph classes has been extensively studied in the literature
(e.g.~\cite{BL02,BGMS14,BELL06,BKM06,BK05,BLM04b,BLM04,BM02,BM03,DGP14,DP14,GR99b,LR04,LR06,LV08,MR99}).
It is straightforward to verify that the class of $H$-free graphs has bounded clique-width if and only if $H$ is an induced subgraph of the 4-vertex path $P_4$
(see also~\cite{DP15}). Hence, Dabrowski and Paulusma~\cite{DP15} investigated for which pairs $(H_1,H_2)$ the class of $(H_1,H_2)$-free graphs has bounded clique-width.
In this paper we solve a number of the open cases.
The underlying research question is:

\medskip
\noindent
{\it What kinds of properties of a graph class ensure that its clique-width is bounded?}

\medskip
\noindent
As such, our paper is to be interpreted as a further step towards this direction.
In particular, we believe there is a clear motivation for our type of research, in which new graph classes of bounded clique-width are identified, because
it may lead to a better understanding of the notion of clique-width. It should be noted that clique-width is one of the most difficult graph parameters to deal with.
To illustrate this, no polynomial-time algorithms are known for computing the clique-width of very restricted graph classes, such as unit interval graphs,
or for deciding whether a graph has clique-width at most $c$ for any fixed $c\geq 4$ (as an aside, such an algorithm does exist for $c=3$~\cite{CHLRR12}).

Rather than coming up with ad hoc techniques for solving specific cases, we aim to develop more general techniques for attacking a number of the open cases simultaneously.
Our technique in this paper is obtained by generalizing an approach followed in the literature. In order to illustrate  this approach with some examples, we first need to
introduce some notation
 (see Section~\ref{s-preliminaries} for all other terminology).

\medskip
\noindent
{\bf Notation.}
The disjoint union $(V(G)\cup V(H), E(G)\cup E(H))$
 of two vertex-disjoint graphs~$G$ and $H$ is denoted by $G+H$ and the disjoint union of $r$ copies of a graph $G$ is denoted by~$rG$.
The  complement of a graph $G$, denoted by $\overline{G}$, has vertex set $V(\overline{G})=V(G)$ and an edge between two distinct vertices
if and only if these vertices are not adjacent in $G$. The graphs $C_r, K_r$ and $P_r$ denote the cycle, complete graph and path on $r$ vertices, respectively.
The graph $\overline{2P_1+P_2}$ is called the {\it diamond}.
The graph $K_{1,3}$ is the 4-vertex star, also called the {\it claw}.
For $1\leq h\leq i\leq j$, let $S_{h,i,j}$ be the {\it subdivided claw} whose three edges are subdivided $h-1$, $i-1$ and $j-1$ times, respectively; note that $S_{1,1,1}=K_{1,3}$.

\medskip
\noindent
{\bf Our technique.}
Dabrowski and Paulusma~\cite{DP14} determined all graphs $H$ for which the class of $H$-free bipartite graphs has bounded clique-width. Such a classification turns out to also be useful
for proving boundedness of the clique-width for other graph classes.
For instance, in order to prove that $(\overline{P_1+P_3},P_1+S_{1,1,2})$-free graphs have bounded clique-width, the given graphs
were first reduced to $(P_1+S_{1,1,2})$-free bipartite graphs~\cite{DP15}.
In a similar way, Dabrowski, Lozin, Raman and Ries~\cite{DLRR12}  proved that $(K_3,K_{1,3}+K_2)$-free graphs and $(K_3,S_{1,1,3})$-free have bounded clique-width by
reducing to a subclass of bipartite graphs.
Note that bipartite graphs are perfect graphs. This motivated us to develop a technique based on perfect graphs that are not necessarily bipartite. In order to so,
we need to combine this approach with an additional tool. This tool is based on the following observation. If the vertex set of a graph can be partitioned into a small number of cliques
and the edges between them are sufficiently sparse, then the clique-width is bounded (see also Lemma~\ref{lem:struct2}).
Our technique can be summarized as follows.

\begin{enumerate}[1.]
\item Reduce the input graph to a graph that is in some subclass of perfect graphs;
\item While doing so, bound the clique covering number of the input graph.
\end{enumerate}
Another well-known subclass of perfect graphs is the class of chordal graphs.
We show that besides the class of bipartite graphs, the class of chordal graphs
and the class of perfect graphs itself may be used for Step~1.\footnote{To exploit this further, we recently worked on a
characterization of the boundedness of clique-width for classes of $H$-free chordal graphs and $H$-free perfect graphs~\cite{BDHP15}.  For this paper, however, we rely only on
existing results from the literature.}
We explain Steps~1-2 of our technique in detail in Section~\ref{s-clique}.

\medskip
\noindent
{\bf Our results.}
In this paper, we investigate whether our technique can be used to find new pairs $(H_1,H_2)$ for which the clique-width of $(H_1,H_2)$-free graphs is bounded.
We show that this is indeed the case. By applying our technique, we are able to present three new classes of
$(H_1,H_2)$-free graphs of bounded clique-width.\footnote{We do not specify our upper bounds as this would complicate our proofs for negligible gain. This is because in our proofs we apply graph operations that exponentially increase the upper bound of the clique-width, which means that the bounds that could be obtained from our proofs would be very large and far from being tight.}
By modifying walls via graph operations that preserve unboundedness of clique-width, we are also able to present a new class
of $(H_1,H_2)$-free graphs of unbounded clique-width. Combining our results leads to the following theorem (see also \figurename~\ref{fig:diamond-graphs}).

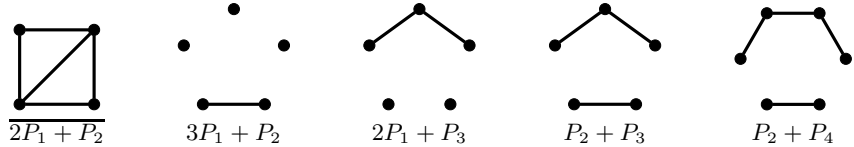
\begin{figure}
\begin{center}
\setlength{\tabcolsep}{1.5em}
\begin{tabular}{ccccc}
\scalebox{0.7}
{\begin{tikzpicture}[scale=1,rotate=45]
\renewcommand*{\EdgeLineWidth}{1.7pt}
\GraphInit[vstyle=Simple]
\SetVertexSimple[MinSize=6pt]
\Vertices{circle}{a,b,c,d}
\Edges(a,b,c,d,a,c)
\end{tikzpicture}}
&
\scalebox{0.7}
{\begin{tikzpicture}[scale=1,rotate=90]
\renewcommand*{\EdgeLineWidth}{1.7pt}
\GraphInit[vstyle=Simple]
\SetVertexSimple[MinSize=6pt]
\Vertices{circle}{a,b,c,d,e}
\Edge(c)(d)
\end{tikzpicture}}
&
\scalebox{0.7}
{\begin{tikzpicture}[scale=1,rotate=90]
\renewcommand*{\EdgeLineWidth}{1.7pt}
\GraphInit[vstyle=Simple]
\SetVertexSimple[MinSize=6pt]
\Vertices{circle}{a,b,c,d,e}
\Edges(e,a,b)
\end{tikzpicture}}
&
\scalebox{0.7}
{\begin{tikzpicture}[scale=1,rotate=90]
\renewcommand*{\EdgeLineWidth}{1.7pt}
\GraphInit[vstyle=Simple]
\SetVertexSimple[MinSize=6pt]
\Vertices{circle}{a,b,c,d,e}
\Edges(e,a,b)
\Edge(c)(d)
\end{tikzpicture}}
&
\scalebox{0.7}
{\begin{tikzpicture}[scale=1,rotate=180]
\renewcommand*{\EdgeLineWidth}{1.7pt}
\GraphInit[vstyle=Simple]
\SetVertexSimple[MinSize=6pt]
\Vertices{circle}{a,b,c,d,e,f}
\Edges(b,c)
\Edges(d,e,f,a)
\end{tikzpicture}}\\
$\overline{2P_1+P_2}$ &
$3P_1+P_2$ &
$2P_1+P_3$ &
$P_2+P_3$ &
$P_2+P_4$
\end{tabular}
\caption{\label{fig:diamond-graphs}The graphs in Theorem~\ref{t-main}.}
\end{center}
\end{figure}

\newpage
\begin{theorem}\label{t-main}
The class of $(H_1,H_2)$-free graphs has bounded clique-width if
\begin{enumerate}[(i)]
\item $H_1= \overline{2P_1+P_2}$ and $H_2=3P_1+P_2$;
\item $H_1=\overline{2P_1+P_2}$ and $H_2=2P_1+P_3$;
\item  $H_1= \overline{2P_1+P_2}$ and $H_2=P_2+P_3$.
\end{enumerate}
The class of $(H_1,H_2)$-free graphs has unbounded clique-width if
\begin{enumerate}[(i)]
\setcounter{enumi}{3}
\item $H_1= \overline{2P_1+P_2}$ and $H_2=P_2+P_4$.
\end{enumerate}
\end{theorem}
We prove statements~(i)--(iv) of Theorem~\ref{t-main} in~Sections~\ref{s-1}--\ref{s-4}, respectively.
In Section~\ref{s-4} we also prove that the {\sc Graph Isomorphism} problem is
{\sc Graph Isomorphism}-complete for the class of $(\overline{2P_1+P_2},P_2+P_4)$-free graphs.
This result was one of the remaining open cases in a line of research initiated by Kratsch and Schweitzer~\cite{KS12}, who tried to classify the complexity of the {\sc Graph Isomorphism} problem in graph classes defined by two forbidden induced subgraphs.
The exact number of open cases is not known, but
Schweitzer~\cite{Sc15} very recently proved that this number is finite.

\medskip
\noindent
{\bf Structural consequences.}
Theorem~\ref{t-main} reduces the number of open cases in the classification of the boundedness of the clique-width for $(H_1,H_2)$-free graphs
to~13 open cases, up to an equivalence relation.\footnote{\label{footnote:equivalent}Let $H_1,\ldots, H_4$ be four graphs. Then the classes of
$(H_1,H_2)$-free graphs and $(H_3,H_4)$-free graphs are {\em
equivalent} if the unordered pair $\{H_3,H_4\}$ can be obtained from the unordered
pair $\{H_1,H_2\}$ by some combination of the following two operations: complementing both graphs in the pair; or if one of the graphs in the pair is $K_3$, replacing it with $\overline{P_1+P_3}$ or vice versa.
If two classes are equivalent then one has bounded clique-width if and only if the other one does (see e.g.~\cite{DP15}).}
Note that the graph $H_1$ is the diamond in each of the four results in Theorem~\ref{t-main}.
Out of the 13 remaining cases, there are still three cases in which~$H_1$ is the diamond, namely when
$H_2\in \{P_1+P_2+P_3, P_1+2P_2, P_1+P_5\}$.
However, for each of these graphs~$H_2$, it is not even known whether the clique-width of the corresponding smaller subclasses of $(K_3,H_2)$-free graphs is bounded. Of particular note is the class of
 $(K_3,P_1+2P_2)$-free graphs, which is contained in all of the above open cases and for which the boundedness of clique-width is unknown. Settling this case is a natural next step in completing the classification.
Note that for $K_3$-free graphs the clique covering number is proportional to the size of the graph.
Another natural research direction is to determine whether the clique-width of $(\overline{P_1+P_4},H_2)$-free
graphs is bounded for $H_2 = P_2+\nobreak P_3$ (the clique-width is known to be unbounded for $H_2 \in \{3P_1+P_2,2P_1+P_3\}$).

Dabrowski, Golovach and Paulusma~\cite{DGP14} showed that {\sc Colouring} restricted to $(sP_1+\nobreak P_2,\allowbreak \overline{tP_1+P_2})$-free graphs is polynomial-time solvable for all pairs of integers~$s,t$.
They justified their algorithm by proving that the clique-width of the class of $(sP_1,\overline{tP_1+P_2})$-free graphs is bounded only for small values of~$s$ and~$t$, namely only for $s\leq 2$ or $t\leq 1$ or $s+t\leq 6$.
In the light of these two results it is natural to try to classify the clique-width of the class of $(sP_1+\nobreak P_2,\allowbreak \overline{tP_1+P_2})$-free graphs for all pairs $(s,t)$.
Theorem~\ref{t-main}, combined with the aforementioned classification of the clique-width of $(sP_1,\overline{tP_1+P_2})$-free graphs and the fact that any class of $(H_1,H_2)$-free graphs has bounded clique-width if and only if the class of $(\overline{H_1},\overline{H_2})$-free graphs has bounded clique-width,
immediately enables us to do this.

\begin{corollary}\label{c-extension}
The class of $(sP_1+P_2,\overline{tP_1+P_2})$-free graphs has bounded clique-width if and only if $s\leq 1$ or $t\leq 1$ or $s+t\leq 5$.
\end{corollary}

\noindent
{\bf Algorithmic consequences.}
Our research was (partially) motivated by a study into the computational complexity of the {\sc Colouring} problem for $(H_1,H_2)$-free
graphs. As mentioned, {\sc Colouring} is polynomial-time solvable on any graph class of bounded clique-width.
Of the three classes for which we prove boundedness of clique-width in this paper, only the case of $(\overline{2P_1+P_2},3P_1+P_2)$-free (and equivalently\textsuperscript{\ref{footnote:equivalent}} $(2P_1+P_2,\overline{3P_1+P_2})$-free) graphs was previously known to be polynomial-time solvable~\cite{DGP14}.
Hence, Theorem~\ref{t-main} gives us four new pairs $(H_1,H_2)$ with the property that
{\sc Colouring} is polynomial-time solvable when restricted to $(H_1,H_2)$-free graphs,
namely if
\begin{itemize}
\item [$\bullet$] $H_1=2P_1+P_2$ and $H_2\in \{\overline{2P_1+P_3},\overline{P_2+P_3}\}$;
\item [$\bullet$] $H_1=\overline{2P_1+P_2}$ and $H_2\in \{2P_1+P_3,P_2+P_3\}$.
\end{itemize}
There are still 15 potential classes of $(H_1,H_2)$-free graphs left for which both the complexity of {\sc Colouring} and the boundedness
of their clique-width is
unknown~\cite{DP15}.

\section{Preliminaries}\label{s-preliminaries}

Below we define some graph terminology used throughout our paper.
For any undefined terminology we refer to Diestel~\cite{Di12}.
Let $G$ be a graph.
For $u\in V(G)$, the set $N(u)=\{v\in V(G)\; |\; uv\in E(G)\}$ is the {\em  neighbourhood} of $u$ in $G$.
The {\em degree} of a vertex in $G$ is the size of its neighbourhood.
The {\em maximum degree} of $G$ is the maximum vertex degree.
For a subset $S\subseteq V(G)$, we let $G[S]$ denote the {\it induced} subgraph of $G$, which has vertex set~$S$ and edge set $\{uv\; |\; u,v\in S, uv\in E(G)\}$.
If $S=\{s_1,\ldots,s_r\}$ then, to simplify notation, we may also write $G[s_1,\ldots,s_r]$ instead of $G[\{s_1,\ldots,s_r\}]$.
Let~$H$ be another graph. We write $H\ssi G$ to indicate that $H$ is an induced subgraph of~$G$.
Let $X\subseteq V(G)$. We write $G\setminus X$ for the graph obtained from~$G$ after removing~$X$.
A set $M\subseteq E(G)$ is a {\it matching} if no two edges in~$M$ share an end-vertex.
We say that two disjoint sets $S\subseteq V(G)$ and $T\subseteq V(G)$  are {\it complete} to each other if every vertex of~$S$ is adjacent to every vertex of~$T$.
If no vertex of $S$ is joined to a vertex of $T$ by an edge, then~$S$ and $T$ are {\it anti-complete} to each other. Similarly, we say that
a vertex $u$ and a set $S$ not containing $u$ may be complete or anti-complete to each other.
Let  $\{H_1,\ldots,H_p\}$ be a set of graphs. Recall that
$G$ is {\it $(H_1,\ldots,H_p)$-free} if $G$ has no induced subgraph isomorphic to a graph in $\{H_1,\ldots,H_p\}$; if $p=1$, we may write $H_1$-free instead of $(H_1)$-free.

The {\em clique-width} of a graph $G$, denoted by $\cw(G)$, is the minimum
number of labels needed to
construct $G$ by
using the following four operations:
\begin{enumerate}[(i)]
\item creating a new graph consisting of a single vertex $v$ with label $i$;
\item taking the disjoint union of two labelled graphs $G_1$ and $G_2$;
\item joining each vertex with label $i$ to each vertex with label $j$ $(i\neq j)$;
\item renaming label $i$ to $j$.
\end{enumerate}
An algebraic term that represents such a construction of $G$ and uses at most~$k$ labels is said to be a {\em $k$-expression} of $G$ (i.e. the clique-width of $G$ is the minimum~$k$ for which $G$ has a $k$-expression).
A class of graphs ${\cal G}$ has \emph{bounded} clique-width if
there is a constant $c$ such that the clique-width of every graph in ${\cal G}$ is at most~$c$; otherwise the clique-width of~${\cal G}$ is
{\em unbounded}.

Let $G$ be a graph.
We say that $G$ is  \emph{bipartite} if its vertex set can be partitioned into two (possibly empty) independent sets $B$ and $W$. We say that $(B,W)$ is a {\em bipartition} of $G$.

Let $G$ be a graph. We define the following two operations.
For an induced subgraph $G'\ssi\nobreak G$,  the {\em subgraph complementation} operation (acting on $G$ with respect to $G'$) replaces every edge present in $G'$
by a non-edge, and vice versa.
Similarly, for two disjoint vertex subsets~$X$ and~$Y$ in~$G$, the {\em bipartite complementation} operation with respect to $X$ and~$Y$ acts on~$G$ by replacing
every edge with one end-vertex in $X$ and the other one in $Y$ by a non-edge and vice versa.

We now state some useful facts for dealing with clique-width. We will use these facts throughout the paper.
Let $k\geq 0$ be a constant and let $\gamma$ be some graph operation.
We say that a graph class ${\cal G'}$ is {\it $(k,\gamma)$-obtained} from a graph class ${\cal G}$
if the following two conditions hold:
\begin{enumerate}[(i)]
\item every graph in ${\cal G'}$ is obtained from a graph in ${\cal G}$ by performing $\gamma$ at most $k$ times, and
\item for every $G\in {\cal G}$ there exists at least one graph
in ${\cal G'}$ obtained from $G$ by performing $\gamma$ at most $k$ times.
\end{enumerate}
We say that $\gamma$ {\em preserves} boundedness of clique-width if
for any finite constant~$k$ and any graph class ${\cal G}$, any graph class
${\cal G}'$ that is $(k,\gamma)$-obtained from ${\cal G}$
has bounded clique-width if and only if  ${\cal G}$  has bounded clique-width.

\begin{enumerate}[{Fact} 1.]
\item \label{fact:del-vert} Vertex deletion preserves boundedness of clique-width~\cite{LR04}.\\[-1em]

\item \label{fact:comp} Subgraph complementation preserves boundedness of clique-width~\cite{KLM09}.\\[-1em]

\item \label{fact:bip} Bipartite complementation preserves boundedness of clique-width~\cite{KLM09}.\\[-1em]

\item \label{fact:deg1} If ${\cal G}$ is a class of graphs and ${\cal G}'$ is the class of graphs obtained from graphs in ${\cal G}$ by recursively deleting all vertices of degree~1, then~${\cal G}$ has bounded clique-width if and only if~${\cal G}'$ has bounded clique-width \cite{BL02,LR04}.
\end{enumerate}

\medskip
\noindent
The following lemmas
are well-known and straightforward to check.

\begin{lemma}\label{lem:tree}
The clique-width of a forest is at most~$3$.
\end{lemma}

\begin{lemma}\label{lem:atmost2}
The clique-width of a graph of maximum degree at most~$2$ is at~most~$4$.
\end{lemma}

\noindent
Let $G$ be a graph. The size of a largest independent set and a largest clique in~$G$ are denoted by $\alpha(G)$ and $\omega(G)$, respectively.
The chromatic number of $G$ is denoted by $\chi(G)$.
We say that~$G$ is {\em perfect} if $\chi(H)=\omega(H)$ for every induced subgraph $H$ of $G$.

We need the following well-known result, due to Chudnovsky, Robertson, Seymour and Thomas.

\begin{theorem}[The Strong Perfect Graph Theorem~\cite{CRST06}]\label{thm:spgt}
A graph is perfect if and only if it is $C_r$-free and $\overline{C_r}$-free
for every odd $r \geq 5$.
\end{theorem}
The {\it clique covering number} $\overline{\chi}(G)$ of a graph $G$ is the smallest number of (mutually vertex-disjoint) cliques such that every vertex of $G$ belongs to exactly one clique.
 If $G$ is perfect, then $\overline{G}$ is also perfect (by Theorem~\ref{thm:spgt}).
By definition,~$\overline{G}$ can be partitioned into $\omega(\overline{G})=\alpha(G)$ independent sets. This leads to the following well-known lemma.

 \begin{lemma}\label{l-clique}
Let $G$ be any perfect graph. Then $\overline{\chi}(G)= \alpha(G)$.
 \end{lemma}
We say that a graph $G$ is \emph{chordal} if it contains no induced cycle on four or more vertices. Bipartite graphs and chordal graphs are perfect (by Theorem~\ref{thm:spgt}).

The following three lemmas give us a number of subclasses of perfect graphs with bounded clique-width. We will make use of these lemmas later on in the proofs as part of our technique.

\begin{lemma}[\cite{DP14}]\label{lem:bipartite}
Let $H$ be a graph. The class of $H$-free bipartite graphs has bounded
clique-width if and only if one of the following cases holds:
\begin{itemize}
\item [$\bullet$] $H=sP_1$ for some $s\geq 1$
\item [$\bullet$] $H\ssi K_{1,3}+3P_1$
\item [$\bullet$] $H\ssi K_{1,3}+P_2$
\item [$\bullet$] $H\ssi P_1+S_{1,1,3}$
\item [$\bullet$] $H\ssi S_{1,2,3}$.
\end{itemize}
\end{lemma}

\begin{lemma}[\cite{GR99b}]\label{lem:diamond-chordal}
The class of chordal $(\overline{2P_1+P_2})$-free graphs has clique-width at most~$3$.
\end{lemma}

\begin{lemma}[\cite{DLRR12}]\label{lem:K3-3P1+P2}
The class of $(K_3,K_{1,3}+P_2)$-free graphs has bounded clique-width.
\end{lemma}

\noindent
Finally, we also need the following lemma, which corresponds to the first lemma of~\cite{DGP14} by complementing the graphs under consideration.

\begin{lemma}[\cite{DGP14}]\label{lem:struct1b}
Let $s\geq 0$ and $t\geq 0$.
Then every $(\overline{sP_1+P_2},tP_1+P_2)$-free graph is $(K_{s+1},tP_1+P_2)$-free or
$(\overline{sP_1+P_2},(s^2(t-1)+2)P_1)$-free.
\end{lemma}

\section{The Clique Covering Lemma}\label{s-clique}

In Section~\ref{s-preliminaries} we stated several lemmas that can be used to bound the clique-width if we can manage to reduce to some specific graph class.
As we shall see, such a reduction is not always sufficient and
the following lemma forms a crucial part of our technique (we use it in the proofs of each of our three main boundedness results).

\begin{lemma}\label{lem:struct2}
Let  $k \geq 1$ be a constant and let $G$ be a $(\overline{2P_1+P_2},2P_2+P_4)$-free graph.
If $\overline{\chi}(G)\leq k$ then $\cw(G)\leq f(k)$ for some function $f$ that only depends on~$k$.
\end{lemma}

\begin{proof}
Let $k\geq 1$.
Suppose $G$ is  a $(\overline{2P_1+P_2},2P_2+P_4)$-free graph with
 $\overline{\chi}(G)\leq k$, that is, $V(G)$ can be partitioned into~$k$
cliques $X_1,\ldots,X_k$.
By Fact~\ref{fact:del-vert}, if any of these cliques
has less than~$k+7$ vertices, we may remove it. If two cliques $X_i,X_j$
are complete to
each other then they can be replaced by the single clique $X_i
\cup X_j$. After doing this exhaustively, we end up with $k' \leq k$ cliques
$Y_1,\ldots,Y_{k'}$, each of which is of size at least $k'+7$ and no two of
which are complete to each other.

Suppose a vertex $x \in Y_i$ has two neighbours $y_1,y_2$ in a different clique
$Y_j$. If~$x$ is non-adjacent to some vertex $y_3 \in Y_j$ then
$G[y_1,y_2,y_3,x]$ is a $\overline{2P_1+P_2}$. Thus~$x$ must be complete to
$Y_j$. If there is another vertex $x' \in Y_i$ which is complete to~$Y_j$, then
every vertex in $Y_j$ has at least two neighbours in $Y_i$, so $Y_i$ and $Y_j$
must be complete to each other, which we assumed was not the case. Therefore,
for any ordered pair $(Y_i,Y_j)$ every vertex of $Y_i$, except possibly one
vertex $x$, has at most one neighbour in $Y_j$. By Fact~\ref{fact:del-vert}, if
such vertices $x$ exist, we may delete them, since there are at most
$k'(k'-1)$ of them. We obtain a set of cliques $Z_1,\ldots,Z_{k'}$, all of
which have size at least $(k'+7)-(k'-1)=8$. Let $G_Z=G[Z_1 \cup \cdots \cup
Z_{k'}]$. We have shown that $G$ has bounded clique-width if and only if $G_Z$
does.

First suppose that $k' \leq 3$. Let $G_Z'$ be the graph obtained from $G_Z$ by
complementing the edges in each set $Z_i$.
As $G_Z'$ has maximum degree at most~2,  it has clique-width at most~4 by Lemma~\ref{lem:atmost2}.
By Fact~\ref{fact:comp}, $G_Z$ has
bounded clique-width if and only if $G_Z'$ does. Hence, $G_Z$, and thus $G$, has bounded clique-width.

Now suppose that $k' \geq 4$.  If $G_Z$ is a union of disjoint
cliques then its clique-width is at most 2. Otherwise, there must be two
vertices in different cliques $Z_i$ that are adjacent. Without loss of
generality, assume $x_6 \in Z_1$ and $x_7 \in Z_2$ are adjacent. We will show
that $G_Z$ (and therefore $G$) contains an induced $2P_2+P_4$, two vertices of
which are $x_6$ and $x_7$. Indeed, since $|Z_1| \geq 8$, there must be a vertex
$x_5 \in Z_1$ that is non-adjacent to $x_7$. Similarly, since $|Z_1| \geq 8$
there must be a vertex $x_8 \in Z_2$ that is non-adjacent to $x_5$ and $x_6$.
Now $G[x_5,x_6,x_7,x_8]$ is a~$P_4$.  Since $|Z_3| \geq 8$, there must be two
vertices $x_3,x_4 \in Z_3$ that are non-adjacent to $x_5, \ldots, x_8$. Since
$|Z_4| \geq 8$, there must be two vertices $x_1,x_2 \in Z_4$ that are
non-adjacent to $x_3, \ldots, x_8$. Now $G[x_1,\ldots,x_8]$ is a $2P_2+P_4$.
This contradiction completes the proof.\qed
\end{proof}
It is easy to see that for any fixed constant $s\geq 2$ we can generalize Lemma~\ref{lem:struct2} to be valid for $(\overline{2P_1+P_2},2K_s+P_4)$-free graphs.
By more complicated arguments it is also possible to generalize it to other graph classes, such as $(\overline{2P_1+P_2},K_s+P_6)$-free graphs for any fixed $s\geq 0$.
However, this is not necessary for the main results of this paper.

\section{The Proof of Theorem~\ref{t-main} (i)}\label{s-1}

Here is the proof of our first main result.

\medskip
\noindent
{\bf Theorem~\ref{t-main} (i).}
{\it The class of $(\overline{2P_1+P_2},3P_1+P_2)$-free graphs has bounded clique-width.}

\begin{proof}
Let $G$ be a $(\overline{2P_1+P_2},3P_1+P_2)$-free graph. Applying
Lemma~\ref{lem:struct1b} we find that $G$ is $(K_3,3P_1+P_2)$-free
or $(\overline{2P_1+P_2},10P_1)$-free.
If $G$ is $(K_3,3P_1+P_2)$-free then it has bounded clique-width by
Lemma~\ref{lem:K3-3P1+P2}, so we may assume it is
$(\overline{2P_1+P_2},10P_1,3P_1+P_2)$-free.

Suppose $G$ contains a $C_5$ (respectively $C_7$) on vertices $v_1,\ldots,v_5$
(respectively $v_1,\ldots,v_7$) in that order. Let $S_i$ be the set of vertices
that have $i$ neighbours on the cycle, but are not on the cycle itself.  Let
$v_i$ and $v_j$ be non-consecutive vertices of the cycle. The set $X$ of
vertices adjacent to both $v_i$ and $v_j$ must be independent, otherwise
$v_i,v_j$ and two adjacent vertices from $X$ would induce a
$\overline{2P_1+P_2}$. Since $G$ is $10P_1$-free, $|X| \leq 9$. Therefore, by
Fact~\ref{fact:del-vert}, we may delete all such vertices, of which there are
at most $9 \times 5 \times 2 \div 2=45$ (respectively $9 \times 7 \times 4 \div
2=126$). All remaining vertices must be adjacent to at most two vertices of the
cycle (so $S_i$ is empty for $i \geq 3$), and if a vertex is adjacent to two
vertices of the cycle, these two vertices must be consecutive vertices of the
cycle.

Suppose $x_1,x_2$ are adjacent to two consecutive vertices of the cycle, $v_i$
and~$v_j$, say. Then $x_1,x_2$ must be adjacent, otherwise $G[v_i,v_j,x_1,x_2]$
would be a $\overline{2P_1+P_2}$. Therefore $S_2$ can be partitioned into at
most five (respectively seven) cliques. Let $Y$ be the set of vertices,
adjacent to $v_1$ and none of the other vertices on the cycle.  If $x_1,x_2 \in
Y$ are non-adjacent then $G[x_1,x_2,v_2,v_4,v_5]$ would be a $3P_1+P_2$, so $Y$
must be a clique. Therefore $S_1$ can be partitioned into at most five
(respectively seven) cliques. Finally, note that if $x_1,x_2 \in S_0$ are
non-adjacent then $G[x_1,x_2,v_1,v_3,v_4]$ is a $3P_1+P_2$, so $S_0$ must be a
clique. By Fact~\ref{fact:del-vert}, we may delete the vertices
$v_1,\ldots,v_5$ (respectively $v_1,\ldots,v_7$).
This leaves a graph whose vertex set can be decomposed into $5+5+1=11$ (respectively
$7+7+1=15$) cliques, in which case we are done by Lemma~\ref{lem:struct2}.

We may therefore assume that $G$ contains no induced $C_5$ or $C_7$. Since $G$
is $(3P_1+P_2)$-free it contains no odd cycle on nine or more vertices. Since it
is $\overline{C_5}$-free (because $\overline{C_5}=C_5$), and
$\overline{2P_1+P_2}$-free, it contains no induced complements of odd cycles of
length 5 or more. By Theorem~\ref{thm:spgt} we find that $G$ must be perfect.
Then~$G$ has clique partition number at most $\alpha(G)$ by Lemma~\ref{l-clique}.
Since $G$ is $10P_1$-free, $\alpha(G) \leq 9$.
Applying Lemma~\ref{lem:struct2} completes the proof.\qed
\end{proof}

\section{The Proof of Theorem~\ref{t-main} (ii)}\label{s-2}

In this section we prove the second of our four main results.

\medskip
\noindent
{\bf Theorem~\ref{t-main} (ii).}
{\it The class of $(\overline{2P_1+P_2},2P_1+P_3)$-free graphs has bounded clique-width.}

\begin{proof}
Let $G$ be a $(\overline{2P_1+P_2},2P_1+P_3)$-free graph.
We need the following claim.

\medskip
\newcounter{ctrclaimone}
\phantomsection\refstepcounter{ctrclaimone}
\noindent
{\it Claim \thectrclaimone.\label{clm:1:1} Let $C$ and $I$ be a clique and independent set of $G$, respectively, with $C\cap I=\emptyset$.
Then there is a set $S \subseteq C\cup I$ containing at most four vertices, such that every edge with one end-vertex in $C$ and the other one in $I$  is incident to at least one vertex of $S$.}

\medskip
\noindent
We prove Claim~\ref{clm:1:1} as follows.
Assume $|I|,|C| \geq 5$, as otherwise we can simply set~$S$ to equal either $I$ or $C$ respectively.
Since $G$ is $\overline{2P_1+P_2}$-free, every vertex in~$I$ must be adjacent to zero, one
or all vertices of $C$. Since $G$ is $\overline{2P_1+P_2}$-free, at most one
vertex $z$ of $I$ can be complete to $C$. If such a vertex $z$ vertex exits,
let $I'=I \setminus\{z\}$, and add $z$ to $S$, otherwise let $I'=I$ and leave $S$ empty. Now $|I'| \geq 4$ and every
vertex of~$I'$ has at most one neighbour in $C$. It remains to show that it is
possible to disconnect~$I'$ and~$C$ by deleting at most three vertices (which we add to $S$).  If a
vertex $x$ in~$C$ has two neighbours and two non-neighbours in $I'$, then these
four vertices, together with~$x$ would induce a $2P_1+P_3$ in $G$.  If some
vertex of~$C$ is adjacent to all but at most one vertex of~$I'$, then since each
vertex of~$I'$ has at most one neighbour in~$C$, deleting at most two vertices
in~$C$ will disconnect $I'$ and~$C$.  We may therefore assume that each vertex
in~$C$ has at most one neighbour in $I'$.  Therefore the edges between~$I'$
and~$C$ form a matching.  If there are no edges between~$C$ and~$I'$ then we are
done.  Suppose $x \in I'$ is adjacent to $y \in C$. Since $|C| \geq 5$, we can
choose $y' \in C$ which is not adjacent to $x$. Since $|I'| \geq 4$, we can
choose $x',x'' \in I'$ which are non-adjacent to~$y$ and $y'$. However, then
$G[x',x'',x,y,y']$ is a $2P_1+P_3$, which is a contradiction. This completes the
proof of Claim~\ref{clm:1:1}.

\medskip
\noindent
Now suppose $G$ contains a $C_4$, say on vertices $v_1,v_2,v_3,v_4$ in order.
Let $X$ be the set of vertices non-adjacent to $v_1,v_2,v_3$ and $v_4$.  For $i
\in \{1,2,3,4\}$ let $W_i$ be the set of vertices adjacent to $v_i$, but
non-adjacent to all other vertices of the cycle.  For $i \in \{1,2\}$ let $V_i$
be the set of vertices not on the cycle that are adjacent to precisely
$v_{i-1}$ and $v_{i+1}$ on the cycle (throughout this part of the proof we
interpret subscripts modulo~4). For $i \in \{1,2,3,4\}$, let $Y_i$ be the
set of vertices adjacent to precisely $v_i$ and $v_{i+1}$ on the cycle. No
vertex can be adjacent to three or more vertices of the cycle, otherwise this
vertex together with three of its neighbours on the cycle would induce a
$\overline{2P_1+P_2}$ in $G$.

If $x,y\in W_i \cup X$ are non-adjacent then $G[x,y,v_{i+1},v_{i+2},v_{i+3}]$ is
a $2P_1+P_3$. Therefore $W_i \cup X$ is a clique. If $x,y \in Y_i$ are
non-adjacent then $G[v_i,v_{i+1},x,y]$ is a $\overline{2P_1+P_2}$. Therefore
$Y_i$ is a clique. If $x,y \in V_i$ are adjacent then $G[x,y,v_{i-1},v_{i+1}]$
is a $\overline{2P_1+P_2}$, so $V_i$ is an independent set. This means that the
vertex set of $G$ can be partitioned into a cycle on four vertices, eight
cliques and two independent sets. By Claim~\ref{clm:1:1},
after deleting the original cycle (four vertices) and at most $4 \times 2
\times 8 = 48$ vertices (which we may do by Fact~\ref{fact:del-vert}), we
obtain a graph whose vertex set is partitioned into eight cliques and two
independent sets such that the two independent sets are not in the same
components as the cliques. The components containing the cliques have bounded
clique-width by Lemma~\ref{lem:struct2}. The two independent sets form a
bipartite $(2P_1+P_3)$-free graph, which has bounded clique-width by
Lemma~\ref{lem:bipartite}. This completes the proof for the case where $G$ contains a
$C_4$.

We may now assume that $G$ is $(C_4,\overline{2P_1+P_2},2P_1+P_3)$-free. Because
$G$ is $(2P_1+P_3)$-free, it cannot contain a cycle on eight or more vertices.
Suppose it contains a cycle on vertices $v_1,\ldots,v_k$ in order, where $k \in
\{5,6,7\}$.  Let $X$ be the set of vertices with no neighbours on the cycle,
$W_i$ be the set of vertices adjacent to $v_i$, but no other vertices on the
cycle, $V_i$ be the set of vertices adjacent to $v_i$ and $v_{i+1}$, but no
other vertices of the cycle and if $v_i$ and $v_j$ are not consecutive vertices
of the cycle, let $V_{i,j}$ be the set of vertices adjacent to both $v_i$ and
$v_j$. (Throughout this part of the proof we interpret subscripts modulo $k$.
Note that a vertex may be in more than one set $V_{i,j}$.)

The set $X \cup W_i$
must be a clique, otherwise two non-adjacent vertices in
$X\cup W_i$ together with
$v_{i+1},v_{i+2},v_{i+3}$ would form a $2P_1+P_3$. The set $V_i$ must be a
clique, as otherwise two non-adjacent vertices in $V_i$, together with $v_i$ and
$v_{i+1}$ would from a $\overline{2P_1+P_2}$. The set $V_{i,j}$ cannot contain
two vertices, otherwise these two vertices, together with $v_i$ and $v_j$, would
form a $C_4$ or a $\overline{2P_1+P_2}$, depending on whether the two vertices
were non-adjacent or adjacent, respectively.
We delete all vertices from all the $V_{i,j}$ sets; we may do so  by
Fact~\ref{fact:del-vert} as there are at most $\frac{1}{2}k(k-3)$ of such vertices.
In this way we obtain a graph that can be partitioned into at
most $2k$ cliques. Therefore $G$ has bounded clique-width by
Lemma~\ref{lem:struct2}.

Finally, we may assume that $G$ contains no induced cycle on four or more
vertices. In other words, we may assume that $G$ is chordal. It remains to recall that
$(\overline{2P_1+P_2})$-free chordal graphs have bounded
clique-width by Lemma~\ref{lem:diamond-chordal}. This completes the proof.\qed
\end{proof}

\section{The Proof of Theorem~\ref{t-main} (iii)}\label{s-3}

In this section we prove the third of our four main results, namely that the class of $(\overline{2P_1+P_2},P_2+P_3)$-free graphs has bounded clique-width.
We first establish, via a series of lemmas, that we may restrict ourselves to graphs in this class that are also $(C_4,C_5,C_6,K_5)$-free.

\begin{lemma}\label{lem:diamond-p2p3-k5}
The class of those $(\overline{2P_1+P_2},P_2+P_3)$-free graphs that contain a~$K_5$ has bounded clique-width.
\end{lemma}

\begin{proof}
Let $G$ be a $(\overline{2P_1+P_2},P_2+P_3)$-free graph.
Let $X$ be a maximal (by set inclusion) clique in $G$ containing at least five
vertices. Since $X$ is maximal and $(\overline{2P_1+P_2})$-free, every vertex not
in $X$ has at most one neighbour in~$X$. By Fact~\ref{fact:deg1} we may
therefore assume that every component of $G \setminus X$ contains at least two
vertices.

Suppose there is a $P_3$ in $G \setminus X$, say on vertices $x_1,x_2,x_3$ in
that order. Since $|X| \geq 5$, we can find $y_1,y_2 \in X$ none of which are
adjacent to any of $x_1,x_2,x_3$. Then $G[y_1,y_2,x_1,x_2,x_3]$ is a $P_2+P_3$.
Hence $G \setminus X$ is $P_3$-free and must therefore be a union of disjoint
cliques $X_1,\ldots,X_k$.  Suppose there is only at most one such clique. Then
$\overline{G}$ is a $(2P_1+P_2)$-free bipartite graph, and so $G$ has bounded
clique-width by Fact~\ref{fact:comp} and Lemma~\ref{lem:bipartite}.
From now on we assume that $k\geq 2$, that is, $G \setminus X$ contains at least two cliques.

Suppose that some vertex $x\in X$ is adjacent to a vertex $y \in X_i$. We claim
that~$x$ can have at most one non-neighbour in any $X_j$.
First suppose $j \neq i$.
For contradiction, assume that $x$ is non-adjacent to $z_1,z_2 \in
X_j$, where $j \neq i$.  Since $|X| \geq 5$ and each vertex that is not in $X$
has at most one neighbour in~$X$, there must be a vertex $x' \in X$ that is
non-adjacent to $y,z_1$ and $z_2$. Then $G[z_1,z_2,x',x,y]$ is a $P_2+P_3$, a contradiction.
Now suppose $j=i$.
Since $k \geq 2$, there must
be another clique $X_j$ with $j \neq i$. Since $X_j$ must contain at least two
vertices and $x$ can have at most one non-neighbour in $X_j$, there must be a
neighbour $y'$ of $x$ in $X_j$. By the same argument as above, $x$ can
therefore have at most one non-neighbour in $X_i$. We conclude that if some
vertex $x$ has a neighbour in $\{X_1,\ldots,X_k\}$ then it has at most one
non-neighbour in each $X_j$.

As every vertex in every $X_i$ has at most one neighbour in $X$, this means
that at most two vertices in $X$ have a neighbour in $X_1\cup \cdots \cup X_k$.
Therefore, by deleting at most two vertices of $X$, we obtain a graph which is
a disjoint union of cliques and therefore has clique-width at most 2. Therefore
by Fact~\ref{fact:del-vert}, the clique-width of $G$ is bounded, which
completes the proof.\qed
\end{proof}

\begin{lemma}\label{lem:diamond-p2p3-c5}
The class of those $(\overline{2P_1+P_2},P_2+P_3,K_5)$-free graphs that contain an induced $C_5$ has
bounded clique-width.
\end{lemma}

\begin{proof}
Let $G$ be a $(\overline{2P_1+P_2},P_2+P_3,K_5)$-free graph containing a
$C_5$, say on vertices $v_1,v_2,v_3,v_4,v_5$ in order. Let $Y$ be
the set of vertices adjacent to $v_1$ and~$v_2$ (and possibly other vertices on
the cycle). If $y_1,y_2 \in Y$ are non-adjacent then $G[v_1,v_2,y_1,y_2]$ would
be a $\overline{2P_1+P_2}$. Therefore $Y$ is a clique. Since $G$ is
$K_5$-free,~$Y$ contains at most four vertices. Therefore by
Fact~\ref{fact:del-vert} we may assume that no vertex in $G$ has two
consecutive neighbours on the cycle.  This also means that no vertex has three
or more neighbours on the cycle.  For $i \in \{1,2,3,4,5\}$,  let~$V_i$ be the
set of vertices not on the cycle that are adjacent to $v_{i-1}$ and $v_{i+1}$,
but non-adjacent to all other vertices of the cycle (subscripts are interpreted
modulo~5 throughout this proof).  Suppose there are two vertices $x,y$, both of
which are adjacent to the same vertex on the cycle, say $v_1$, and non-adjacent
to all other vertices of the cycle. If $x$ and $y$ are adjacent, then
$G[x,y,v_2,v_3,v_4]$ is a $P_2+P_3$, otherwise $G[v_3,v_4,x,v_1,y]$ is a
$P_2+P_3$. This contradiction means that there is at most one vertex whose only
neighbour on the cycle is $v_1$. By Fact~\ref{fact:del-vert}, we may therefore
assume that there is no vertex with exactly one neighbour on the cycle. Let $X$
be the set of vertices with no neighbours on the cycle. Note that every vertex
not on the cycle is either in $X$ or in some set $V_i$.

Now $X$ must be an independent set, since if two vertices in $x_1,x_2 \in X$
are adjacent, then $G[x_1,x_2,v_1,v_2,v_3]$ would induce a $P_2+P_3$ in $G$.
Also, $V_i$ must be an independent set, since if $x,y \in V_i$ are adjacent
then $G[x,y,v_{i-1},v_{i+1}]$ is a $\overline{2P_1+P_2}$.

We say that two sets $V_i$ and $V_j$
are \emph{consecutive} (respectively \emph{opposite}) if $v_i$ and~$v_j$ are
distinct adjacent (respectively non-adjacent) vertices of the cycle.
We say that a set $X$ or $V_i$ is \emph{large} if it contains at least three
vertices, otherwise it is \emph{small}.
We say that a bipartite graph with bipartition classes $A$ and $B$ is a {\it matching} ({\it co-matching})
if every vertex in $A$ has at most one neighbour (non-neighbour) in~$B$, and vice versa.

We now prove a series of claims about the edges between these sets.

\begin{enumerate}
\item \emph{$G[V_i \cup X]$ is a matching.} Indeed if some vertex $x$ in $V_i$
(respectively $X$) is adjacent to two vertices $y_1,y_2$ in $X$ (respectively
$V_i$), then $G[v_{i+2},v_{i+3},y_1,x,y_2]$ is a $P_2+P_3$.

\smallskip
\item \emph{If $V_i$ and $V_j$ are opposite then $G[V_i\cup V_j]$ is a matching.}
Suppose for contradiction that $x \in V_1$ is adjacent to two vertices $y,y'
\in V_3$.  Then $G[v_2,x,y,y']$ would be a $\overline{2P_1+P_2}$, a
contradiction.

\smallskip
\item \emph{If $V_i$ and $V_j$ are consecutive then $G[V_i \cup V_j]$ is a
co-matching.} Suppose for contradiction that $x \in V_1$ is non-adjacent to two
vertices $y,y' \in V_2$. Then $G[x,v_5,y,v_3,y']$ is a $P_2+P_3$, a
contradiction.

\smallskip
\item \emph{If $V_i$ is large then $X$ is anti-complete to $V_{i-2} \cup
V_{i+2}$.} Suppose for contradiction that $V_3$ is large and $x \in X$ has a
neighbour $y \in V_1$.  Then since $V_3$ is large and both $G[X \cup V_3]$ and
$G[V_1 \cup V_3]$ are matchings, there must be a vertex $z \in V_3$ that is
non-adjacent to both $x$ and $y$. Then $G[x,y,v_3,v_4,z]$ is a $P_2+P_3$, a
contradiction.

\smallskip
\item \emph{If $V_i$ is large then $V_{i-1}$ is anti-complete to $V_{i+1}$.}
Suppose for contradiction that $V_2$ is large and $x \in V_1$ has a neighbour
$y \in V_3$. Since $V_2$ is large and each vertex in $V_1 \cup V_3$ has at most
one non-neighbour in $V_2$, there must be a vertex $z \in V_2$ that is adjacent
to both $x$ and $y$. Now $G[x,y,v_2,z]$ is a $\overline{2P_1+P_2}$, a
contradiction.

\smallskip
\item \emph{If $V_{i-1},V_{i},V_{i+1}$ are large then $V_i$ is complete to
$V_{i-1} \cup V_{i+1}$.} Suppose for contradiction that $V_1,V_2,V_3$ are large
and some vertex $x \in V_1$ is non-adjacent to a vertex $y \in V_2$.  Since
$V_3$ is large and $G[V_2 \cup V_3]$ is a co-matching, there must be two
vertices $z,z' \in V_3$, adjacent to $y$. By the previous claim, since~$V_2$ is
large, $z,z'$ must be non-adjacent to $x$.  Therefore $G[x,v_5,z,y,z']$ is a
$P_2+\nobreak P_3$, which is a contradiction.
\end{enumerate}

By Fact~\ref{fact:del-vert} we may delete the vertices
$v_1,\ldots, v_5$
and all vertices in every small set~$X$ or $V_i$. Let $G'$ be the
graph obtained from the resulting graph by complementing the edges between any
two consecutive $V_i,V_j$. By Fact~\ref{fact:bip}, $G'$ has bounded
clique-width if and only if $G$ does. If at most three of $V_1,\ldots,V_5,X$
are large, then $G'$ has maximum degree at most~2 and we are done by
Lemma~\ref{lem:atmost2}.
We may therefore assume that at
least
four of $V_1,\ldots,V_5,X$ are large, so at
least three of $V_1,\ldots,V_5$ are large.

First suppose there is an edge in $G$ between a vertex in $X$ and a vertex in~$V_i$ for some $i$.
Then $V_{i-2},V_{i+2}$ must be small
(and as such we already removed them).
Consequently, $V_{i-1},V_i,V_{i+1}$ must be large.
However, in this case, every large~$V_j$ is either complete or anti-complete to every other large $V_{j'}$ in $G$ and~$X$
is anti-complete  to $V_{i-1} \cup V_{i+1}$ in $G$. Therefore $G'$ has maximum
degree at most~1
implying that $G'$, and thus $G$, has bounded clique-width by Lemma~\ref{lem:atmost2}.

Now suppose that there are no edges in $G$ between any vertex in $X$ and any vertex in $V_i$ for all $i$. Since $X$ is an independent set, every vertex in $X$ forms a component in $G$ of size~1. We can therefore delete every vertex in $X$ without affecting the clique-width of $G$. That is, in this case we may assume that $X$ is not large.
In this case, as stated above, we may assume that at least four of $V_1,\ldots,V_5$ are large. We may without loss
of generality assume that these sets are $V_1,\ldots,V_4$, whereas $V_5$ may or may not be large.
If $V_5$ is large, then every large $V_i$ is either complete or anti-complete to every other large
$V_j$ in $G$. If $V_5$ is small (and as such not in $G'$) then the same holds
 with the possible exception of~$V_1$ and $V_4$. Hence $G'$ has
maximum degree at most~1 implying that $G'$, and thus~$G$, has bounded clique-width by Lemma~\ref{lem:atmost2}.
This completes the proof.\qed
\end{proof}

\begin{lemma}\label{lem:diamond-p2p3-c4}
The class of those $(\overline{2P_1+P_2},P_2+P_3,K_5,C_5)$-free graphs that contain an induced $C_4$ has
bounded clique-width.
\end{lemma}

\begin{proof}
Suppose that $G$ is a $(\overline{2P_1+P_2},P_2+P_3,K_5,C_5)$-free graph
containing a~$C_4$, say on vertices $v_1,v_2,v_3,v_4$ in order.  Let $Y$ be the
set of vertices adjacent to~$v_1$ and $v_2$ (and possibly other vertices on the
cycle). If $y_1,y_2 \in Y$ are non-adjacent then $G[v_1,v_2,y_1,y_2]$ would be
a $\overline{2P_1+P_2}$. Therefore $Y$ is a clique. Since~$G$ is $K_5$-free, there are at most four such
vertices.  Therefore by Fact~\ref{fact:del-vert} we may assume that no vertex
in $G$ has two consecutive neighbours on the cycle.  For $i \in \{1,2\}$ let
$V_i$ be the set of vertices outside the cycle adjacent to $v_{i+1}$ and~$v_{i+3}$ (where $v_5=v_1$).
For $i \in \{1,2,3,4\}$ let $W_i$ be the set of vertices whose unique neighbour
on the cycle is $v_i$. Let $X$ be the set of vertices with no neighbours on the
cycle.

We first prove the following properties:
\newcounter{propctr}
\renewcommand{\thepropctr}{(\roman{propctr})}
\begin{enumerate}[(i)]
\item [\phantomsection\refstepcounter{propctr}\thepropctr\label{prop:vi-indep}]  $V_i$ are independent sets for $i=1,2$.
\item [\phantomsection\refstepcounter{propctr}\thepropctr\label{prop:wi-indep}] $W_i$ are independent sets for $i=1,2,3,4$.
\item [\phantomsection\refstepcounter{propctr}\thepropctr\label{prop:x-indep}] $X$ is an independent set.
\item [\phantomsection\refstepcounter{propctr}\thepropctr\label{prop:x-wi-anti}] $X$ is anti-complete to $W_i$ for $i=1,2,3,4$.
\item [\phantomsection\refstepcounter{propctr}\thepropctr\label{prop:w3-empty}] Without loss of generality $W_3=\emptyset$ and $W_4=\emptyset$.
\item [\phantomsection\refstepcounter{propctr}\thepropctr\label{prop:w1-w2-anti}] Without loss of generality $W_1$ is anti-complete to $W_2$.
\end{enumerate}
To prove Property~\ref{prop:vi-indep}, if $x,y \in V_i$ are adjacent then $G[x,y,v_{i+1},v_{i+3}]$ is a  $\overline{2P_1+P_2}$.
For $i=1,\ldots,4$, the set $W_i \cup X$ must also be independent, since if $x,y \in W_1 \cup X$ were adjacent
then $G[x,y,v_2,v_3,v_4]$ would be a $P_2+P_3$. This proves Properties~\ref{prop:wi-indep}--\ref{prop:x-wi-anti}.

To prove Property~\ref{prop:w3-empty}, suppose that $x\in W_1$ and $y \in W_3$ are adjacent. In that case $G[v_1,v_2,v_3,y,x]$
would be a $C_5$. This contradiction means that no vertex of~$W_1$ is adjacent
to a vertex of $W_3$. Now suppose that $x,x' \in W_1$ and $y \in W_3$. Then
$G[y,v_3,x,v_1,x']$ would be a $P_2+P_3$ by Property~\ref{prop:wi-indep}. Therefore, if both~$W_1$ and $W_3$ are
non-empty, then they each contain at most one vertex and we can delete these
vertices by Fact~\ref{fact:del-vert}. Without loss of generality we may
therefore assume that $W_3$ is empty. Similarly, we may assume $W_4$ is empty. Hence we have shown Property~\ref{prop:w3-empty}.

We are left to prove Property~\ref{prop:w1-w2-anti}.
Suppose that $x \in W_1$ is adjacent to $y \in W_2$. Then~$x$ cannot have a
neighbour in~$V_2$. Indeed, suppose for contradiction that $x$ has a neighbour
$z \in V_2$. Then $G[x,z,y,v_1]$ is a $\overline{2P_1+P_2}$ if $y$ and $z$ are
adjacent, and $G[x,y,v_2,v_3,z]$ is a $C_5$ if $y$ and $z$ are not adjacent.
By symmetry,~$y$ cannot have a neighbour in $V_1$.
Now $y$ must be complete to~$V_2$.  Indeed, if $y$ has a non-neighbour $z \in
V_2$ then $G[x,y,z,v_3,v_4]$ is a $P_2+P_3$. By symmetry, $x$ is complete to~$V_1$.
Recall that $W_1\cup X$ is an independent set by Properties~\ref{prop:wi-indep}--\ref{prop:x-wi-anti}.
We conclude that any vertex in $W_1$ with a neighbour in $W_2$ is complete to $V_1$ and anti-complete to $V_2\cup X$.
Similarly, any vertex in $W_2$ with a neighbour in~$W_1$ is complete to $V_2$ and anti-complete to
$V_1\cup X$.

Let $W_1^*$ (respectively $W_2^*$) be the set of vertices in $W_1$ (respectively $W_2$) that have a neighbour in $W_2$ (respectively $W_1$).
Then, by Fact~\ref{fact:bip}, we may apply two bipartite complementations, one
between $W_1^*$ and $V_1 \cup \{v_1\}$ and the other between~$W_2^*$ and $V_2 \cup \{v_2\}$.
After these operations, $G$ will be split into two disjoint
parts, $G[W_1^* \cup W_2^*]$ and $G \setminus (W_1^* \cup W_2^*)$, both of
which are induced subgraphs of~$G$. The first of these is a bipartite
$(P_2+P_3)$-free graph and therefore has bounded clique-width  by
Lemma~\ref{lem:bipartite}. We therefore only need to consider the second graph
$G \setminus (W_1^* \cup W_2^*)$. In other words, we may assume without loss of generality that~$W_1$ is
anti-complete to $W_2$. This proves Property~\ref{prop:w1-w2-anti}.

\medskip
\noindent
If a vertex in $X$ has no neighbours in $V_1 \cup V_2$ then it is an isolated
vertex by Property~\ref{prop:x-wi-anti} and the definition of the set $X$.
In this case we may delete it without affecting the clique-width.
Hence, we may assume without loss of generality that every vertex in $X$ has at least one neighbour in $V_1\cup V_2$.
We partition $X$ into three sets $X_0,X_1,X_2$ as follows.
 Let~$X_1$
(respectively $X_2$) denote the set of vertices in~$X$ with at least one
neighbour in~$V_1$ (respectively $V_2$), but no neighbours in~$V_2$
(respectively $V_1$). Let $X_0$ denote the set of vertices in $X$ adjacent to
at least one vertex of $V_1$ and at least one vertex of $V_2$.

Let $G^*=G[V_1 \cup V_2 \cup W_1 \cup W_2 \cup X_1 \cup X_2]$.
We prove the following additional properties:

\begin{enumerate}[(i)]
\item [\phantomsection\refstepcounter{propctr}\thepropctr\label{prop:g-star-bip}] $G^*$ is bipartite.
\item [\phantomsection\refstepcounter{propctr}\thepropctr\label{prop:x_0-empty}] Without loss of generality $X_0\neq \emptyset$.
\item [\phantomsection\refstepcounter{propctr}\thepropctr\label{prop:v_1-xnbr-v2-comp}] Every vertex in $V_1$ that has a neighbour in $X$ is complete to $V_2$.
\item [\phantomsection\refstepcounter{propctr}\thepropctr\label{prop:v_2-xnbr-v1-comp}]  Every vertex in $V_2$ that has a neighbour in $X$ is complete to $V_1$.
\item [\phantomsection\refstepcounter{propctr}\thepropctr\label{prop:x_0-one-v_i-nbr}] Every vertex in $X_0$ has exactly one neighbour in $V_1$ and exactly one neighbour in $V_2$.
\item [\phantomsection\refstepcounter{propctr}\thepropctr\label{prop:v_i-one-x_0-nbr}] Without loss of generality, every vertex in $V_1\cup V_2$ has at most one neighbour in $X_0$.
\item [\phantomsection\refstepcounter{propctr}\thepropctr\label{prop:v_1-w_2-anti}] Without loss of generality, $V_1$ is anti-complete to $W_2$.
\item [\phantomsection\refstepcounter{propctr}\thepropctr\label{prop:v_2-w_1-anti}] Without loss of generality, $V_2$ is anti-complete to $W_1$.
\end{enumerate}
Property~\ref{prop:g-star-bip} can be seen has follows.
Because $G$ is $(P_2+P_3,C_5)$-free, $G^*$ has no induced odd
cycles of length at least 5. Suppose, for contradiction, that $G^*$ is not
bipartite.  Then it must contain an induced $C_3$.  Now $V_1, V_2, W_1, W_2,
X_1$ and~$X_2$ are independent sets, so at most one vertex of the $C_3$ can be
in any one of these sets. The set $X_1$ is anti-complete to $V_2, W_1, W_2$  and~$X_2$
(by definition of $V_2$ and Properties~\ref{prop:x-indep} and~\ref{prop:x-wi-anti}).  Hence no vertex of the~$C_3$ can be in~$X_1$. Similarly, no vertex of the
$C_3$ be be in $X_2$. The sets $W_1$ and $W_2$ are anti-complete to each other by Property~\ref{prop:w1-w2-anti},
so the $C_3$ must therefore consist of one vertex from each of~$V_1$ and~$V_2$,
along with one vertex from either $W_1$ or $W_2$. However, in this case, these
three vertices, along with either $v_1$ or $v_2$, respectively would induce a
$\overline{2P_1+P_2}$ in $G$, which would be a contradiction.  Hence we have proven Property~\ref{prop:g-star-bip}.

We now prove Property~\ref{prop:x_0-empty}. Suppose $X_0$ is empty.
Then, since $G^*$ is $(P_2+P_3)$-free and bipartite (by Property~\ref{prop:g-star-bip}), it has bounded clique-width by
Lemma~\ref{lem:bipartite}. Hence,~$G$ has bounded clique-width by Fact~\ref{fact:del-vert}, since we
may delete $v_1,v_2,v_3$ and~$v_4$ to obtain $G^*$. This proves Property~\ref{prop:x_0-empty}.

We now prove Property~\ref{prop:v_1-xnbr-v2-comp}.
Let $y_1 \in V_1$ have a neighbour $x \in X$. Suppose, for contradiction, that $y_1$ has a
non-neighbour $y_2 \in V_2$. Then $G[x,y_2,v_1,v_2,y_1]$ is a~$C_5$ if $x$ is
adjacent to~$y_2$ and $G[x,y_1,v_1,y_2,v_3]$ is a $P_2+P_3$ if $x$ is
non-adjacent to $y_2$, a contradiction. This proves Property~\ref{prop:v_1-xnbr-v2-comp}. By symmetry, Property~\ref{prop:v_2-xnbr-v1-comp} holds.

We now prove Property~\ref{prop:x_0-one-v_i-nbr}. By definition, every vertex in $X_0$ has at least one neighbour in $V_1$ and at least one neighbour in $V_2$.
Suppose, for contradiction, that a vertex $x \in X_0$ has two neighbours $y,y' \in V_1$. By definition, $x$ must also have a
neighbour $z \in V_2$. Then $z$ must be adjacent to both
$y$ and $y'$ by Property~\ref{prop:v_2-xnbr-v1-comp}. However, then $G[x,z,y,y']$ is a $\overline{2P_1+P_2}$ by Property~\ref{prop:vi-indep}, a
contradiction.This proves Property~\ref{prop:x_0-one-v_i-nbr}.

We now prove Property~\ref{prop:v_i-one-x_0-nbr}.
Suppose a vertex $y \in V_1$ has two neighbours $x,x' \in X_0$. If there is
another vertex $z \in X_0$ then $z$ must have a unique neighbour~$z'$ in
$V_1$. If $z'$ is a different vertex from $y$ then $G[z,z',x,y,x']$ would be a
$P_2+P_3$ by Properties~\ref{prop:vi-indep} and~\ref{prop:x-indep}.  Thus $z'=y$, that is, every vertex in~$X_0$ must be adjacent to $y$
and to no other vertex of~$V_1$.
By Fact~\ref{fact:del-vert}, we may delete~$y$. In the resulting graph no
vertex of $X$ would have neighbours in both $V_1$ and~$V_2$. So $X_0$ would
become empty, in which case we can argue as
in the proof of Property~\ref{prop:x_0-empty}. This proves Property~\ref{prop:v_i-one-x_0-nbr}.

We now prove Property~\ref{prop:v_1-w_2-anti}.
First, for $i \in \{1,2\}$, suppose that a vertex $y \in V_i$ is adjacent to a vertex $x \in X$. Then $y$ can have at most one non-neighbour in $W_i$. Indeed, suppose
for contradiction that $z,z' \in W_i$ are non-neighbours of~$y$. Then
$G[x,y,z,v_i,z']$ is a $P_2+P_3$ by Properties~\ref{prop:wi-indep} and~\ref{prop:w1-w2-anti}, a contradiction.
We claim that at most one vertex of $W_2$ has a neighbour in $V_1$. Suppose, for contradiction, that $W_2$ contains two vertices $w$ and $w'$
adjacent to (not necessarily distinct) vertices $z$ and $z'$ in $V_1$,
respectively. Since $X_0\neq \emptyset$ by Property~\ref{prop:x_0-empty}, there must be a vertex $y
\in V_2$ with a neighbour in $X_0$. As we just showed that such a vertex~$y$ can have at most one
non-neighbour in $W_2$, we may assume without loss of generality that~$y$ is
adjacent to $w$. Since $y$ has a neighbour in $X$, it must also be adjacent to
$z$ by Property~\ref{prop:v_2-xnbr-v1-comp}. Now $G[w,z,y,v_2]$ is a $\overline{2P_1+P_2}$, which is a contradiction.
Therefore at most one vertex of $W_2$ has a neighbour in~$V_1$ and similarly,
at most one vertex of $W_1$ has a neighbour in $V_2$. By
Fact~\ref{fact:del-vert}, we may delete these vertices if they exist. This proves Properties~\ref{prop:v_1-w_2-anti} and~\ref{prop:v_2-w_1-anti}.

\medskip
\noindent
For $i=1,2$ let $V_i'$ be the set of vertices in $V_i$ that have a neighbour in $X_0$. We
show two more properties:

\begin{enumerate}[(i)]
\item [\phantomsection\refstepcounter{propctr}\thepropctr\label{prop:w_1x_1-unique-v_1-nbr}] Every vertex in $W_1 \cup X_1$ is adjacent to either none, precisely one or all vertices of $V_1'$.
\item [\phantomsection\refstepcounter{propctr}\thepropctr\label{prop:w_2x_2-unique-v_2-nbr}] Every vertex of $W_2 \cup X_2$ is adjacent to either none, precisely one or all vertices of $V_2'$.
\end{enumerate}
We prove Property~\ref{prop:w_1x_1-unique-v_1-nbr} as follows. Suppose a vertex $x\in X_1 \cup W_1$ has at least two neighbours in $z,z' \in V_1$.
We claim that $x$ must be complete to $V_1'$. Suppose, for contradiction, that $x
$ is not adjacent to $y \in
V_1'$. By definition, $y$ has a neighbour $y' \in X_0$. Then $G[y,y',z,x,z']$
is a $P_2+\nobreak P_3$ by Properties~\ref{prop:vi-indep},~\ref{prop:x-indep} and~\ref{prop:x-wi-anti}, a contradiction. This proves Property~\ref{prop:w_1x_1-unique-v_1-nbr}. Property~\ref{prop:w_2x_2-unique-v_2-nbr} follows by symmetry.

\medskip
\noindent
Let $W_i'$ and $X_i'$ be the sets of vertices in $W_i$ and
$X_i$ respectively that are adjacent to precisely one vertex of $V_i'$.
We delete $v_1,v_2,v_3$ and $v_4$, which we may do by Fact~\ref{fact:del-vert}.
We do a bipartite complementation between $V_1'$ and those vertices in $W_1\cup X_1$ that are complete to $V_1'$.
We also do this between $V_2'$ and those vertices in $W_2\cup X_2$ that are complete to $V_2'$.
Finally, we perform a bipartite complementation between $V_1'$ and $V_2\setminus V_2'$ and also between $V_2'$ and $V_1\setminus V_1'$.
We may do all of this by Fact~\ref{fact:bip}.
Afterwards, Properties~\ref{prop:vi-indep}--\ref{prop:w1-w2-anti}, \ref{prop:v_1-xnbr-v2-comp}, \ref{prop:v_2-xnbr-v1-comp}, \ref{prop:v_1-w_2-anti}--\ref{prop:w_2x_2-unique-v_2-nbr} and the definitions of~$V_1'$, $V_2'$, $W_1'$, $W_2'$, $X_1$, $X_2$ imply that there are no edges between
the following two vertex-disjoint graphs:
\begin{enumerate}[1.]
\item $G[W_1' \cup W_2' \cup X_1' \cup X_2' \cup V_1' \cup V_2' \cup X_0]$ and
\item $G \setminus (W_1' \cup W_2' \cup X_1' \cup X_2' \cup V_1' \cup V_2' \cup X_0 \cup
\{v_1,v_2,v_3,v_4\})$
\end{enumerate}
Both of these graphs are induced subgraphs of $G$. The second
of these graphs does not contain any vertices of $X_0$. So it is bipartite by Property~\ref{prop:g-star-bip} and
therefore has bounded clique-width, as argued before
(in the proof of Property~\ref{prop:x_0-empty}).

Now consider the first graph, which is $G[W_1' \cup W_2' \cup X_1' \cup X_2' \cup V_1' \cup
V_2' \cup X_0]$.  By Fact~\ref{fact:bip}, we may complement the edges between
$V_1'$ and $V_2'$. This yields a new graph~$G'$.
By definition of $V_1', V_2'$ and Properties~\ref{prop:v_1-xnbr-v2-comp} and~\ref{prop:v_2-xnbr-v1-comp}, we find that $V_1'$ is anti-complete to $V_2'$ in $G'$.
Hence, by definition of $V_1', V_2'$ and Properties~\ref{prop:vi-indep},~\ref{prop:x-indep}, \ref{prop:x_0-one-v_i-nbr} and~\ref{prop:v_i-one-x_0-nbr},
we find that $G'[V_1' \cup V_2' \cup X_0]$
is a disjoint union of $P_3$'s. For $i \in \{1,2\}$, every vertex in $W_i' \cup
X_i'$ is adjacent to precisely one vertex in $V_i'$ by definition. As the last bipartite complementation operation did not affect these sets, this is still the case in
$G'$.
By Properties~\ref{prop:wi-indep}--\ref{prop:x-wi-anti} and~\ref{prop:w1-w2-anti}, we find that $W_1'\cup W_2'\cup X_0\cup X_1'\cup X_2'$ is an independent set.
Then, by also using Properties~\ref{prop:v_1-w_2-anti} and~\ref{prop:v_2-w_1-anti} together with the definitions of $X_1$ and $X_2$,  we find that
no vertex in $W_i' \cup X_i'$ has any other neighbour in $G'$ besides its neighbour in $V_i'$.
Therefore~$G'$ is a disjoint union of trees and thus has bounded clique-width
by Lemma~\ref{lem:tree}.  We conclude that $G$ has bounded clique-width. This completes the proof of Lemma~\ref{lem:diamond-p2p3-c4}.\qed
\end{proof}

\begin{lemma}\label{lem:diamond-p2p3-c6}
The class of those $(\overline{2P_1+P_2},P_2+P_3,K_5,C_5,C_4)$-free graphs that contain an induced $C_6$ has
bounded clique-width.
\end{lemma}

\begin{proof}
Let $G$ be a $(\overline{2P_1+P_2},P_2+P_3,K_5,C_5,C_4)$-free graph
containing a $C_6$, say on vertices $v_1,v_2,v_3,v_4,v_5,v_6$ in order.
Let $Y$ be the set of vertices adjacent to $v_1$ and $v_2$ (and possibly other
vertices on the cycle). If $y_1,y_2 \in Y$ are non-adjacent then
$G[v_1,v_2,y_1,y_2]$ would be a $\overline{2P_1+P_2}$.  Therefore $Y$ must be a
clique. Since $G$ is $K_5$-free, $Y$ contains at most four vertices.  Therefore
by Fact~\ref{fact:del-vert} we may assume that no vertex in $G$ has two
consecutive neighbours on the cycle. Suppose there are two vertices $x$ and
$x'$, both of which are adjacent to two non-consecutive vertices of the cycle
$v_i$ and $v_j$. Then if $x$ and $x'$ are adjacent, $G[x,x',v_i,v_j]$ would be
a $\overline{2P_1+P_2}$, otherwise $G[x,v_i,x',v_j]$ would be a $C_4$, a
contradiction. Thus for every two non-adjacent vertices on the cycle, there can
be at most one vertex adjacent to both of them. By Fact~\ref{fact:del-vert} we
may delete all such vertices. We conclude that every other vertex which is not
on the cycle can be adjacent to at most one vertex on the cycle. Suppose $x$ is
adjacent to $v_1$, but not $v_2,v_3,v_4,v_5,v_6$. Then $G[x,v_1,v_3,v_4,v_5]$
would be a $P_2+P_3$. Therefore no vertex which is not on the cycle can have a
neighbour on the cycle. If two vertices $x$ and $x'$ are not adjacent to any
vertex of the cycle then they cannot be adjacent, otherwise
$G[x,x',v_1,v_2,v_3]$ would be a $P_2+P_3$. Therefore the remaining graph is
composed of a $C_6$ and
zero or more isolated vertices. Hence, $G$ has bounded
clique-width. This completes the proof.\qed
\end{proof}
\noindent
We now use Lemmas~\ref{lem:diamond-p2p3-k5}--\ref{lem:diamond-p2p3-c6} and the fact
that $(\overline{2P_1+P_2},P_2+P_3,C_4,C_5,C_6)$-free graphs are chordal graphs, and so have bounded clique-width by Lemma~\ref{lem:diamond-chordal},
to obtain:

\medskip
\noindent
{\bf Theorem~\ref{t-main} (iii).}
{\it The class of $(\overline{2P_1+P_2},P_2+P_3)$-free graphs has bounded clique-width.}

\begin{proof}
Suppose $G$ is a $(\overline{2P_1+P_2},P_2+P_3)$-free graph.
By Lemmas~\ref{lem:diamond-p2p3-k5}--\ref{lem:diamond-p2p3-c6}, we may assume that $G$ is
$(\overline{2P_1+P_2},P_2+P_3,K_5,C_5,C_4,C_6)$-free.
Because $G$ is $(P_2+P_3)$-free, it contains no induced cycles of length 7 or more.
Hence $G$ is chordal, that is, it is a $(\overline{2P_1+P_2})$-free chordal graph, in which case
the clique-width of $G$ is bounded by Lemma~\ref{lem:diamond-chordal}. This completes the proof of the theorem.\qed
\end{proof}

\section{The Proof of Theorem~\ref{t-main} (iv)}\label{s-4}

To prove our fourth main result we need the well-known notion of a {\em wall}. We do not formally define this notion but instead refer to
\figurename~\ref{f-walls}, in which three examples of walls of different height are depicted.

\begin{figure}
\begin{center}
\begin{minipage}{0.2\textwidth}
\centering
\begin{tikzpicture}[scale=0.4, every node/.style={scale=0.3}]
\GraphInit[vstyle=Simple]
\SetVertexSimple[MinSize=6pt]
\Vertex[x=1,y=0]{v10}
\Vertex[x=2,y=0]{v20}
\Vertex[x=3,y=0]{v30}
\Vertex[x=4,y=0]{v40}
\Vertex[x=5,y=0]{v50}

\Vertex[x=0,y=1]{v01}
\Vertex[x=1,y=1]{v11}
\Vertex[x=2,y=1]{v21}
\Vertex[x=3,y=1]{v31}
\Vertex[x=4,y=1]{v41}
\Vertex[x=5,y=1]{v51}

\Vertex[x=0,y=2]{v02}
\Vertex[x=1,y=2]{v12}
\Vertex[x=2,y=2]{v22}
\Vertex[x=3,y=2]{v32}
\Vertex[x=4,y=2]{v42}

\Edges(    v10,v20,v30,v40,v50)
\Edges(v01,v11,v21,v31,v41,v51)
\Edges(v02,v12,v22,v32,v42)

\Edge(v01)(v02)

\Edge(v10)(v11)

\Edge(v21)(v22)

\Edge(v30)(v31)

\Edge(v41)(v42)

\Edge(v50)(v51)

\end{tikzpicture}
\end{minipage}
\begin{minipage}{0.3\textwidth}
\centering
\begin{tikzpicture}[scale=0.4, every node/.style={scale=0.3}]
\GraphInit[vstyle=Simple]
\SetVertexSimple[MinSize=6pt]
\Vertex[x=1,y=0]{v10}
\Vertex[x=2,y=0]{v20}
\Vertex[x=3,y=0]{v30}
\Vertex[x=4,y=0]{v40}
\Vertex[x=5,y=0]{v50}
\Vertex[x=6,y=0]{v60}
\Vertex[x=7,y=0]{v70}

\Vertex[x=0,y=1]{v01}
\Vertex[x=1,y=1]{v11}
\Vertex[x=2,y=1]{v21}
\Vertex[x=3,y=1]{v31}
\Vertex[x=4,y=1]{v41}
\Vertex[x=5,y=1]{v51}
\Vertex[x=6,y=1]{v61}
\Vertex[x=7,y=1]{v71}

\Vertex[x=0,y=2]{v02}
\Vertex[x=1,y=2]{v12}
\Vertex[x=2,y=2]{v22}
\Vertex[x=3,y=2]{v32}
\Vertex[x=4,y=2]{v42}
\Vertex[x=5,y=2]{v52}
\Vertex[x=6,y=2]{v62}
\Vertex[x=7,y=2]{v72}

\Vertex[x=1,y=3]{v13}
\Vertex[x=2,y=3]{v23}
\Vertex[x=3,y=3]{v33}
\Vertex[x=4,y=3]{v43}
\Vertex[x=5,y=3]{v53}
\Vertex[x=6,y=3]{v63}
\Vertex[x=7,y=3]{v73}

\Edges(    v10,v20,v30,v40,v50,v60,v70)
\Edges(v01,v11,v21,v31,v41,v51,v61,v71)
\Edges(v02,v12,v22,v32,v42,v52,v62,v72)
\Edges(    v13,v23,v33,v43,v53,v63,v73)

\Edge(v01)(v02)

\Edge(v10)(v11)
\Edge(v12)(v13)

\Edge(v21)(v22)

\Edge(v30)(v31)
\Edge(v32)(v33)

\Edge(v41)(v42)

\Edge(v50)(v51)
\Edge(v52)(v53)

\Edge(v61)(v62)

\Edge(v70)(v71)
\Edge(v72)(v73)
\end{tikzpicture}
\end{minipage}
\begin{minipage}{0.35\textwidth}
\centering
\begin{tikzpicture}[scale=0.4, every node/.style={scale=0.3}]
\GraphInit[vstyle=Simple]
\SetVertexSimple[MinSize=6pt]
\Vertex[x=1,y=0]{v10}
\Vertex[x=2,y=0]{v20}
\Vertex[x=3,y=0]{v30}
\Vertex[x=4,y=0]{v40}
\Vertex[x=5,y=0]{v50}
\Vertex[x=6,y=0]{v60}
\Vertex[x=7,y=0]{v70}
\Vertex[x=8,y=0]{v80}
\Vertex[x=9,y=0]{v90}

\Vertex[x=0,y=1]{v01}
\Vertex[x=1,y=1]{v11}
\Vertex[x=2,y=1]{v21}
\Vertex[x=3,y=1]{v31}
\Vertex[x=4,y=1]{v41}
\Vertex[x=5,y=1]{v51}
\Vertex[x=6,y=1]{v61}
\Vertex[x=7,y=1]{v71}
\Vertex[x=8,y=1]{v81}
\Vertex[x=9,y=1]{v91}

\Vertex[x=0,y=2]{v02}
\Vertex[x=1,y=2]{v12}
\Vertex[x=2,y=2]{v22}
\Vertex[x=3,y=2]{v32}
\Vertex[x=4,y=2]{v42}
\Vertex[x=5,y=2]{v52}
\Vertex[x=6,y=2]{v62}
\Vertex[x=7,y=2]{v72}
\Vertex[x=8,y=2]{v82}
\Vertex[x=9,y=2]{v92}

\Vertex[x=0,y=3]{v03}
\Vertex[x=1,y=3]{v13}
\Vertex[x=2,y=3]{v23}
\Vertex[x=3,y=3]{v33}
\Vertex[x=4,y=3]{v43}
\Vertex[x=5,y=3]{v53}
\Vertex[x=6,y=3]{v63}
\Vertex[x=7,y=3]{v73}
\Vertex[x=8,y=3]{v83}
\Vertex[x=9,y=3]{v93}

\Vertex[x=0,y=4]{v04}
\Vertex[x=1,y=4]{v14}
\Vertex[x=2,y=4]{v24}
\Vertex[x=3,y=4]{v34}
\Vertex[x=4,y=4]{v44}
\Vertex[x=5,y=4]{v54}
\Vertex[x=6,y=4]{v64}
\Vertex[x=7,y=4]{v74}
\Vertex[x=8,y=4]{v84}

\Edges(    v10,v20,v30,v40,v50,v60,v70,v80,v90)
\Edges(v01,v11,v21,v31,v41,v51,v61,v71,v81,v91)
\Edges(v02,v12,v22,v32,v42,v52,v62,v72,v82,v92)
\Edges(v03,v13,v23,v33,v43,v53,v63,v73,v83,v93)
\Edges(v04,v14,v24,v34,v44,v54,v64,v74,v84)

\Edge(v01)(v02)
\Edge(v03)(v04)

\Edge(v10)(v11)
\Edge(v12)(v13)

\Edge(v21)(v22)
\Edge(v23)(v24)

\Edge(v30)(v31)
\Edge(v32)(v33)

\Edge(v41)(v42)
\Edge(v43)(v44)

\Edge(v50)(v51)
\Edge(v52)(v53)

\Edge(v61)(v62)
\Edge(v63)(v64)

\Edge(v70)(v71)
\Edge(v72)(v73)

\Edge(v81)(v82)
\Edge(v83)(v84)

\Edge(v90)(v91)
\Edge(v92)(v93)
\end{tikzpicture}
\end{minipage}
\caption{Walls of height 2, 3, and 4, respectively.}\label{f-walls}
\end{center}
\end{figure}
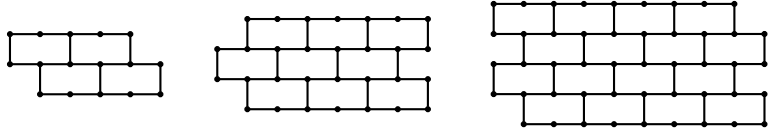

The class of walls is well known to have unbounded clique-width; see for example~\cite{KLM09}. We need a more general result.
The \emph{subdivision} of an edge $uv$ in a graph replaces $uv$ by a new vertex~$w$ with edges $uw$ and~$vw$.
A \emph{$k$-subdivided wall} is a graph obtained from a wall after subdividing each edge
exactly $k$ times for some constant $k\geq 0$. The following lemma is well known.

\begin{lemma}[\cite{LR06}]\label{l-walls}
For any constant $k\geq 0$, the class of $k$-subdivided walls has unbounded clique-width.
\end{lemma}

\medskip
\noindent
{\bf Theorem~\ref{t-main} (iv).}
{\it The class of $(\overline{2P_1+P_2},P_2+P_4)$-free graphs has unbounded clique-width.}

\begin{proof}
Let $n \geq 2$ and let $G_n$ be a wall of height $n$. Note that $G_n$ is a
connected bipartite graph. Let $A$ and $C$ be its two bipartition classes. We
subdivide every edge in $G_n$ exactly once to obtain a 1-subdivided wall. Let
$B$ be the set of  new vertices introduced by this operation. We then apply a
bipartite complementation between $A$ and $C$, which results in a graph~$G_n'$.
The set of graphs $\{G_n'\}_{n \geq 2}$ has unbounded clique-width by
Lemma~\ref{l-walls} and Fact~\ref{fact:bip}. Hence it suffices to prove
that~$G_n'$ is $(\overline{2P_1+P_2},P_2+P_4)$-free.  We do this using three
observations.
\begin{enumerate}[(i)]
\item \label{proof-observations} $A$ and $C$ are independent sets in $G_n'$ that are complete to
each-other, in other words, $G_n'[A\cup C]$ is a complete bipartite graph.
\item $B$ is an independent set and every vertex of $B$ has exactly one
neighbour in $A$ and exactly one neighbour in $C$.
\item No two vertices of~$B$ have the same neighbourhood.
\end{enumerate}

We now prove that $G_n'$ is $\overline{2P_1+P_2}$-free. For contradiction,
suppose that $G_n'$ contains an induced subgraph $H$ isomorphic to
$\overline{2P_1+P_2}$.  Since $G_n'[A \cup C]$ is complete bipartite, any
triangle in $G_n'$ must contain a vertex of $B$. Since the vertices of $B$ have
degree~2, this means that the two degree-2 vertices of $H$ must be in~$B$. As
$G_n'[A\cup C]$ is complete bipartite, one of the degree-3 vertices of~$H$ is
in~$A$ and the other one is in~$C$.  This implies that the two degree-2
vertices in $H$ have the same neighbourhood. Since both of these vertices
belong to~$B$, this is a contradiction.

It remains to prove that $G_n'$ is $(P_2+P_4)$-free.  For contradiction,
suppose that~$G_n'$ contains an induced subgraph $H$ isomorphic to $P_2+P_4$.
Let $H_1$ and $H_2$ be the connected components of $H$ isomorphic to $P_2$ and
$P_4$, respectively.  Since $G_n'[A\cup C]$ is complete bipartite, $H_2$ must
contain at least one vertex of~$B$. Since the two neighbours of any vertex of
$B$ are adjacent,  any vertex of $B$ in~$H_2$ must be an end-vertex of $H_2$.
Then, as $A$ and $C$ are independent sets, $H_2$ contains a vertex of both $A$
and $C$. As $H_1$ can contain at most one vertex of $B$ (because $B$ is an
independent set), $H_1$ contains a vertex~$u\in A\cup C$. However, $G_n'[A\cup
C]$ is complete bipartite and $H_2$ contains a vertex of both $A$ and $C$.
Hence,~$u$ has a neighbour in $H_2$, which is not possible.  This completes the
proof of Theorem~\ref{t-main}~(iv).\qed
\end{proof}

We finish this section with one more result.  A {\it dominating} vertex in a
graph~$G$ is a vertex adjacent to all other vertices of $G$.  We need the
following two well-known observations (see e.g.~\cite{KS12}).

\begin{lemma}\label{l-universal}
Let $G_1'$ and $G_2'$ be the graphs obtained from two graphs $G_1$ and $G_2$,
respectively, by adding a dominating vertex. Then $G_1'$ and $G_2'$ are
isomorphic if and only if $G_1$ and $G_2$ are.
\end{lemma}

\begin{lemma}\label{l-subdivide}
Let $G_1'$ and $G_2'$ be the graphs obtained from subdividing every edge of two
graphs $G_1$ and $G_2$, respectively, exactly once.  Then $G_1'$ and $G_2'$ are isomorphic if
and only if $G_1$ and $G_2$ are.
\end{lemma}

\begin{theorem}\label{t-iso}
{\sc Graph Isomorphism} is {\sc Graph Isomorphism}-complete for the class of
$(\overline{2P_1+P_2},P_2+P_4)$-free graphs.
\end{theorem}

\begin{proof}
Let $G_1$ and $G_2$ be arbitrary graphs.  For $i=1,2$ we modify $G_i$ as
follows.  First, add four dominating vertices. (Note that these added vertices
are pairwise adjacent.) This ensures that the graph has minimum degree at
least~3.  Let $A_i$ be the set of vertices in the resulting graph. Subdivide
every edge once and let $C_i$ be the set of new vertices. Note that this
results in a bipartite graph with bipartition classes $A_i$ and $C_i$.
Subdivide each edge in this modified graph and let $B_i$ be the set of new
vertices. Call the resulting graph $G_i'$. Finally, apply a bipartite
complementation between $A_i$ and $C_i$. Let $G_i''$ be the resulting graph.
Now in the graph $G_i''$, the sets of vertices $A_i,B_i$ and $C_i$ satisfy the
three \hyperref[proof-observations]{observations} from the proof of
Theorem~\ref{t-main}~(iv) and $G_1''$ and $G_2''$ are therefore
$(\overline{2P_1+P_2},P_2+P_4)$-free by exactly the same arguments.

We claim that $G_1$ and $G_2$ are isomorphic if and only if $G_1''$ and $G_2''$
are. In order to see this, we first use Lemmas~\ref{l-universal}
and~\ref{l-subdivide} to deduce that $G_1$ and $G_2$ are isomorphic if and only
if $G_1'$ and $G_2'$ are. It remains to show that $G_1'$ and $G_2'$ are
isomorphic if and only if $G_1''$ and $G_2''$ are. Note that for $i=1,2$, every
vertex in $A_i$ has degree at least~3 in both $G_i'$ and $G_i''$, every vertex
of $B_i$ has degree exactly~2 in both $G_i'$ and $G_i''$ and every vertex of
$C_i$ has degree exactly~2 in $G_i'$ and degree at least~3 in $G_i''$.
Now a vertex is in $B_i$ if and only if it is adjacent to a vertex of degree at
least~3 in $G_i'$ if and only if it is of degree exactly~2 in $G_i''$. A vertex
in $G_i'$ or $G_i''$ is in $A_i$ if and only if it is adjacent to at least
three vertices of degree~2.  Hence, every isomorphism from $G_1'$ to $G_2'$ and
every isomorphism from~$G_1''$ and $G_2''$ maps the vertices of $A_1, B_1$ and
$C_1$ to the vertices of $A_2, B_2$ and $C_2$, respectively.  The claim follows
since for $i=1,2$ the graph $G_i''$ is obtained from~$G_i'$ by adding all edges
between $A_i$ and $C_i$.\qed
\end{proof}

\bibliographystyle{abbrv}

\bibliography{mybib}

\begin{thebibliography}{10}

\bibitem{BL02}
R.~Boliac and V.~V. Lozin.
\newblock On the clique-width of graphs in hereditary classes.
\newblock {\em Proc. ISAAC 2002, LNCS}, 2518:44--54, 2002.

\bibitem{BGMS14}
F.~Bonomo, L.~N. Grippo, M.~Milani\v{c}, and M.~D. Safe.
\newblock Graphs of power-bounded clique-width.
\newblock {\em arXiv}, abs/1402.2135, 2014.

\bibitem{BDHP15}
A.~Brandst{\"a}dt, K.~K. Dabrowski, S.~Huang, and D.~Paulusma.
\newblock Bounding the clique-width of {$H$-free} chordal graphs.
\newblock {\em manuscript}, 2015.

\bibitem{BELL06}
A.~Brandst{\"a}dt, J.~Engelfriet, H.-O. Le, and V.~V. Lozin.
\newblock Clique-width for 4-vertex forbidden subgraphs.
\newblock {\em Theory of Computing Systems}, 39(4):561--590, 2006.

\bibitem{BKM06}
A.~Brandst{\"a}dt, T.~Klembt, and S.~Mahfud.
\newblock {$P_6$}- and triangle-free graphs revisited: structure and bounded
  clique-width.
\newblock {\em Discrete Mathematics and Theoretical Computer Science},
  8(1):173--188, 2006.

\bibitem{BK05}
A.~Brandst{\"a}dt and D.~Kratsch.
\newblock On the structure of ({$P_5$},gem)-free graphs.
\newblock {\em Discrete Applied Mathematics}, 145(2):155--166, 2005.

\bibitem{BLM04b}
A.~Brandst{\"a}dt, H.-O. Le, and R.~Mosca.
\newblock Gem- and co-gem-free graphs have bounded clique-width.
\newblock {\em International Journal of Foundations of Computer Science},
  15(1):163--185, 2004.

\bibitem{BLM04}
A.~Brandst{\"a}dt, H.-O. Le, and R.~Mosca.
\newblock Chordal co-gem-free and ({$P_5$},gem)-free graphs have bounded
  clique-width.
\newblock {\em Discrete Applied Mathematics}, 145(2):232--241, 2005.

\bibitem{BM02}
A.~Brandst{\"a}dt and S.~Mahfud.
\newblock Maximum weight stable set on graphs without claw and co-claw (and
  similar graph classes) can be solved in linear time.
\newblock {\em Information Processing Letters}, 84(5):251--259, 2002.

\bibitem{BM03}
A.~Brandst{\"a}dt and R.~Mosca.
\newblock On variations of {$P_4$}-sparse graphs.
\newblock {\em Discrete Applied Mathematics}, 129(2--3):521--532, 2003.

\bibitem{CRST06}
M.~Chudnovsky, N.~Robertson, P.~Seymour, and R.~Thomas.
\newblock The strong perfect graph theorem.
\newblock {\em Annals of Mathematics}, 164:51--229, 2006.

\bibitem{CHLRR12}
D.~G. Corneil, M.~Habib, J.-M. Lanlignel, B.~A. Reed, and U.~Rotics.
\newblock Polynomial-time recognition of clique-width $\leq3$ graphs.
\newblock {\em Discrete Applied Mathematics}, 160(6):834--865, 2012.

\bibitem{CR05}
D.~G. Corneil and U.~Rotics.
\newblock On the relationship between clique-width and treewidth.
\newblock {\em SIAM Journal on Computing}, 34:825--847, 2005.

\bibitem{Co14}
B.~Courcelle.
\newblock Clique-width and edge contraction.
\newblock {\em Information Processing Letters}, 114(1--2):42--44, 2014.

\bibitem{CMR00}
B.~Courcelle, J.~A. Makowsky, and U.~Rotics.
\newblock Linear time solvable optimization problems on graphs of bounded
  clique-width.
\newblock {\em Theory of Computing Systems}, 33(2):125--150, 2000.

\bibitem{DGP14}
K.~K. Dabrowski, P.~A. Golovach, and D.~Paulusma.
\newblock Colouring of graphs with {Ramsey-type} forbidden subgraphs.
\newblock {\em Theoretical Computer Science}, 522:34--43, 2014.

\bibitem{DLRR12}
K.~K. Dabrowski, V.~V. Lozin, R.~Raman, and B.~Ries.
\newblock Colouring vertices of triangle-free graphs without forests.
\newblock {\em Discrete Mathematics}, 312(7):1372--1385, 2012.

\bibitem{DP14}
K.~K. Dabrowski and D.~Paulusma.
\newblock Classifying the clique-width of ${H}$-free bipartite graphs.
\newblock {\em Proc. COCOON 2014, LNCS}, 8591:489--500, 2014.

\bibitem{DP15}
K.~K. Dabrowski and D.~Paulusma.
\newblock Clique-width of graph classes defined by two forbidden induced
  subgraphs.
\newblock {\em Proc. CIAC 2015, LNCS}, to appear.

\bibitem{Di12}
R.~Diestel.
\newblock {\em Graph Theory, 4th Edition}, volume 173 of {\em Graduate texts in
  mathematics}.
\newblock Springer, 2012.

\bibitem{FRRS09}
M.~R. Fellows, F.~A. Rosamond, U.~Rotics, and S.~Szeider.
\newblock Clique-width is {NP-Complete}.
\newblock {\em SIAM Journal on Discrete Mathematics}, 23(2):909--939, 2009.

\bibitem{GR99b}
M.~C. Golumbic and U.~Rotics.
\newblock On the clique-width of some perfect graph classes.
\newblock {\em International Journal of Foundations of Computer Science},
  11(03):423--443, 2000.

\bibitem{Gu07}
F.~Gurski.
\newblock Graph operations on clique-width bounded graphs.
\newblock {\em CoRR}, abs/cs/0701185, 2007.

\bibitem{Johansson98}
{\"O}.~Johansson.
\newblock Clique-decomposition, {NLC}-decomposition, and modular decomposition
  - relationships and results for random graphs.
\newblock {\em Congressus Numerantium}, 132:39--60, 1998.

\bibitem{KLM09}
M.~Kamiński, V.~V. Lozin, and M.~Milanič.
\newblock Recent developments on graphs of bounded clique-width.
\newblock {\em Discrete Applied Mathematics}, 157(12):2747--2761, 2009.

\bibitem{KR03b}
D.~Kobler and U.~Rotics.
\newblock Edge dominating set and colorings on graphs with fixed clique-width.
\newblock {\em Discrete Applied Mathematics}, 126(2--3):197--221, 2003.

\bibitem{KS12}
S.~Kratsch and P.~Schweitzer.
\newblock Graph isomorphism for graph classes characterized by two forbidden
  induced subgraphs.
\newblock {\em Proc. WG 2012, LNCS}, 7551:34--45, 2012.

\bibitem{LR04}
V.~V. Lozin and D.~Rautenbach.
\newblock On the band-, tree-, and clique-width of graphs with bounded vertex
  degree.
\newblock {\em SIAM Journal on Discrete Mathematics}, 18(1):195--206, 2004.

\bibitem{LR06}
V.~V. Lozin and D.~Rautenbach.
\newblock The tree- and clique-width of bipartite graphs in special classes.
\newblock {\em Australasian Journal of Combinatorics}, 34:57--67, 2006.

\bibitem{LV08}
V.~V. Lozin and J.~Volz.
\newblock The clique-width of bipartite graphs in monogenic classes.
\newblock {\em International Journal of Foundations of Computer Science},
  19(02):477--494, 2008.

\bibitem{MR99}
J.~A. Makowsky and U.~Rotics.
\newblock On the clique-width of graphs with few {$P_4$}'s.
\newblock {\em International Journal of Foundations of Computer Science},
  10(03):329--348, 1999.

\bibitem{Oum08}
S.-I. Oum.
\newblock Approximating rank-width and clique-width quickly.
\newblock {\em ACM Transactions on Algorithms}, 5(1):10, 2008.

\bibitem{OS06}
S.-I. Oum and P.~D. Seymour.
\newblock Approximating clique-width and branch-width.
\newblock {\em Journal of Combinatorial Theory, Series B}, 96(4):514--528,
  2006.

\bibitem{Ra07}
M.~Rao.
\newblock {MSOL} partitioning problems on graphs of bounded treewidth and
  clique-width.
\newblock {\em Theoretical Computer Science}, 377(1--3):260--267, 2007.

\bibitem{Sc15}
P.~Schweitzer.
\newblock Towards an isomorphism dichotomy for hereditary graph classes.
\newblock {\em Proc. STACS 2015, LIPIcs}, to appear.

\end{thebibliography}

\end{document}